\documentclass[%
 reprint,
superscriptaddress,
nofootinbib,
longbibliography,
 amsmath,amssymb,
 aps,
prb,
]{revtex4-1}

\usepackage{mathbbol,amsmath,amsfonts,amssymb,bm,dsfont}
\usepackage{amsfonts}
\usepackage{dsfont}
\usepackage{graphicx,xcolor}% Include figure files
\usepackage{dcolumn}% Align table columns on decimal point
\usepackage{bm}% bold math
\usepackage{hyperref}% add hypertext capabilities
\usepackage[mathlines]{lineno}% Enable numbering of text and display math
\usepackage{mathtools}
\usepackage{physics}
\usepackage[normalem]{ulem}
\usepackage{subfigure}
\usepackage{multirow}
\usepackage{amsthm}
\usepackage[title]{appendix}

\providecommand{\ignore}[1]{}

\newif\ifcmnt
\cmnttrue
\ifdefined\cmntsoff\cmntfalse\fi

\ifcmnt
    \providecommand{\aucmnt}[1]{#1}

\else
    \providecommand{\aucmnt}[1]{}

\fi

\newcommand{\bA}{{\mathbb{A}}}

\newcommand{\cA}{{\cal A}}

\newcommand{\cC}{{\cal C}}
\newcommand{\cE}{{\cal E}}

\newcommand{\cH}{{\cal H}}

\newcommand{\cP}{{\cal P}}
\newcommand{\cR}{{\cal R}}
\newcommand{\cS}{{\cal S}}

\newtheorem{proposition}{Proposition}
\newtheorem*{theorem*}{Theorem}
\newtheorem{theorem}{Theorem}
\newtheorem{corollary}{Corollary}[theorem]
\newtheorem{lemma}[theorem]{Lemma}

\begin{document}

%\title{Self-testing Majorana Fermions Outline}
\title{Certified Quantum Measurement of Majorana Fermions}

\author{Abu Ashik Md.~Irfan}
\thanks{These authors contributed equally to this work.}
 \affiliation{Department of Physics, Indiana University, Bloomington, IN
 47405-7105,USA}

\author{Karl Mayer}
\thanks{These authors contributed equally to this work.}
\affiliation{National Institute of Standards and Technology, Boulder, Colorado,
 USA}
\affiliation{Department of Physics, University of Colorado, Boulder, Colorado, USA}

\author{Gerardo Ortiz}
\affiliation{Department of Physics, Indiana University, Bloomington, IN
47405-7105,USA}

\author{Emanuel Knill}
\affiliation{National Institute of Standards and Technology, Boulder, Colorado, USA}
\affiliation{Center for Theory of Quantum Matter, University of Colorado, Boulder,
Colorado, USA}

\date{\today}

\begin{abstract}
We present a quantum self-testing protocol to certify measurements of fermion parity involving Majorana fermion modes.
We show that observing a set of ideal measurement statistics implies
anti-commutativity of the implemented Majorana fermion parity operators, a necessary prerequisite for Majorana detection. Our protocol is robust to experimental errors. We obtain lower bounds on the fidelities of the state and measurement operators that are linear in the errors.
We propose to analyze experimental outcomes in terms of a 
contextuality witness $W$,
which satisfies $\expval{W}\le 3$ for any classical probabilistic
model of the data.
A violation of the inequality witnesses quantum 
contextuality, and the closeness to the maximum ideal value $\langle W \rangle=5$ indicates the degree of confidence in the 
detection of Majorana fermions.

% We show that observing a set of ideal measurement statistics
% implies

\end{abstract}

\maketitle

\section{Introduction}\label{sec:level1}

Topological qubits offer promising basic units for quantum information processing due to their inherent resilience against decoherence \cite{Alicea2012}. 
Majorana fermions \cite{Majorana1937} are candidates for realizing such topological
qubits, and the ability to braid them  is the focus of several recent investigations. 
Theoretically, Majorana fermions emerge from the interplay between 
the existence of a topologically non-trivial vacuum and a, typically,
symmetry-protected physical boundary (or defect). They are realized as zero-energy modes or quasi-particle
excitations of certain quantum systems.
Recent experimental efforts to detect and control Majorana zero-energy modes in
topological superconducting nanowires provide a step towards  realizing 
non-Abelian braiding and, thus, topological computation. Several experimental 
groups have reported evidence of Majorana zero-energy modes,  
such as an observation of a zero bias conductance peak or Shapiro steps 
in superconducting nanowires \cite{Mourik2012,Rokhinson2012}. The evidence,
however, remains indirect and it is unclear what would constitute proof of the existence
of Majorana fermions \cite{Liu2012}.  Moreover, interpretation of what embodies
a Majorana excitation, and its physical realization, in a closed particle number-conserving many-body topological superfluid
deepens the mystery \cite{ortiz2014,ortiz2016}. 

Even if one had strong evidence that a system is in a topological superfluid phase  with emerging
Majorana fermions, in order to reap the advantages of the topological approach to
quantum computing, one must be confident that the measurements performed actually 
implement ideal quantum operations with high fidelity.
This is especially important for proposals where
gates are performed by parity measurements and anyonic 
teleportation, rather than physical braiding \cite{Bonderson2009}.
In this
paper, we present a protocol to certify quantum measurements of observables and 
states using only the statistics of measurements outcomes, while making no 
assumptions about the underlying physics in the experimental apparatus. Our 
technique represents an extension of what is known as \textit{self-testing} in
quantum information \cite{Mayers2004,Acin2007,Colbeck_2011}. In
particular, we are interested in currently proposed platforms utilizing fermionic parity
measurements \cite{Karzig2017}. In this way, and given experimental data, one hopes
to argue for the consistency of that data with the existence of Majorana
fermions. 

In the quantum information literature, self-testing refers to the action 
of uniquely determining a quantum state, up to a certain notion of equivalence. Unlike tomography, self-testing is
based solely on the statistics of 
measurement outcomes, with minimal assumptions about the measurement operators. 
These quantum self-testing protocols
are more stringent than the well-known  Bell tests \cite{Popescu1992}. While violation
of a Bell inequality for a bipartite system establishes that its quantum state is
entangled, it cannot certify, for instance, that its quantum state is maximally
entangled \cite{Romito2017}. Self-testing protocols typically assume 
that the physical system has a 
Hilbert space with a natural local tensor product structure. For self-testing a fermionic system,
however, we have to relax this assumption. In our scenario, involving 6 Majorana fermion modes and 6  parity operators, a minimal assumption is 
compatibility of observables sharing no common Majorana mode. A successful
certification implies that the experimentally measured observables anti-commute
exactly the way ideal fermionic parity operators should. 
We demonstrate that ideal statistics imply {\it emergence} of an invariant four-dimensional tensor-product subspace (encoding two logical qubits) out of a putative Majorana fermion non-tensor-product state space, and the ideal state is a
Bell state up to local unitary equivalence. An observation of the ideal statistics in
our protocol would constitute substantive evidence of the existence of Majorana fermions. This is so, since ideal statistics implies anti-commutativity of a Majorana fermion and its parity operator, a definite smoking gun for Majorana fermion detection.  Experiments, however, suffer from imperfections, and any practical certification protocol should 
include the effect of non-ideal quantum measurement devices and procedures. 
 We have obtained lower bounds on state and operator fidelities, linear in the error, that constitute rigorous statements on robustness of the self-testing protocol for detection of Majorana fermions. 

The paper is organized as follows. Preliminary background concepts and strategy for self-testing Majorana fermion parities are discussed in Sec. \ref{sec:background}. In particular, in Sec. \ref{subsection: Majorana fermion}, we
map Majorana fermion parity operators to two-qubit Pauli operators,
and construct maximal sets of compatible measurements, so-called contexts. 
In Sec.~\ref{subsection: self testing}, we introduce
the notion of quantum self-testing.
Section \ref{sec:statements} describes our particular
measurement scenario and contains a summary of
our main results, which are two theorems, proved later in Sections \ref{sec: Rigidity} and \ref{sec: Robustness}. Specifically, we prove {\it rigidity}  of the measurement scenario in Sec. \ref{sec: Rigidity}, and
address the {\it robustness} 
to small experimental errors  in Sec. \ref{sec: Robustness}.  Finally, in Sec.~\ref{Discussion} we summarize main findings and analyze 
our fermion parity certification protocol from the standpoint of 
a contextuality witness $W$.  We suggest a possible experimental setup and propose to analyze experimental data validating a contextuality inequality involving such $W$. We also emphasize the generality of our approach and its potential application to other quantum measurements involving phenomena such as braiding.  
An accessible discussion, addressed to experimentalists, of what 
an ideal statistics situation means in the context of self-testing fermion parities is presented in Appendix \ref{Appendix: Ideal statistics}. Several technical details, important to appreciate the mathematical and physical implications of our results, are included in the 
the Appendices \ref{Appendix: Jordan}, \ref{Appendix: state fidelity}, and 
\ref{Appendix: equation to derive 3rd column}.

\section{Background}\label{sec:background}

\subsection{Majorana Fermions} \label{subsection: Majorana fermion}

\begin{table}[htb]
    \centering
\begin{tabular}{ |c|c|c|c| } 
\hline
{\sf Logical Qubits} & {\sf Fermion Parities} & {\sf Physical Qubits} \\
\hline
\multirow{2}{*}{ $\ket{00}$} & $\ket{+,+,+}$ & $\ket{\downarrow}\ket{\Phi_-}$ \\ \cline{2-3}
& $\ket{-,+,+}$ & $\ket{\downarrow}\ket{\Phi'_-}$ \\ 
\hline
\multirow{2}{*}{ $\ket{01}$} & $\ket{-,+,-}$ & $-\ket{\downarrow}\ket{\Phi_+}$ \\ \cline{2-3}
& $\ket{+,+,-}$ & $-\ket{\downarrow}\ket{\Phi'_+}$ \\ 
\hline
\multirow{2}{*}{ $\ket{10}$} & $\ket{-,-,+}$ & $\ket{\uparrow}\ket{\Phi'_-}$ \\ \cline{2-3}
& $\ket{+,-,+}$ & $-\ket{\uparrow}\ket{\Phi_-}$ \\ 
\hline
\multirow{2}{*}{ $\ket{11}$} & $\ket{+,-,-}$ & $-\ket{\uparrow}\ket{\Phi'_+}$ \\ \cline{2-3}
& $\ket{-,-,-}$ & $\ket{\uparrow}\ket{\Phi_+}$ \\ 
\hline
\end{tabular}
\caption{Mapping between the (logical) $4$-dimensional and (physical) $8$-dimensional spaces (fermion parity assingnments for $P_{36}$, $P_{12}$ and $P_{45}$ and three qubits representations). For each logical state the upper row corresponds to even parity, while the lower to odd parity. Here $\ket{\Phi_\pm}=\frac{\ket{\uparrow\uparrow}\pm \ket{\downarrow\downarrow}}{\sqrt{2}}$ and $\ket{\Phi'_\pm}=\frac{\ket{\uparrow\downarrow}\pm \ket{\downarrow\uparrow}}{\sqrt{2}}$. }
\label{Isomorphism: 2qubit-3qubit}
\end{table}

Majorana fermion modes are potential blueprint qubits for topological computation. 
Consider $6$ Majorana modes belonging to $6$ different quantum wires or 
vortices. 
Those modes are defined by Majorana operators
$\gamma_{j}^{\;}$ for $j=1,\ldots, 6$, which satisfy the  Majorana algebra
\begin{eqnarray} \notag
\gamma_{j}^{\dagger}=\gamma_{j}^{\;} \ ,  \mbox{ and } \{\gamma_{j}^{\;},\gamma_{k}^{\;}\}=
\gamma_{j}^{\;}\gamma_{k}^{\;}+\gamma_{k}^{\;}\gamma_{j}^{\;}=2\delta_{jk} .
\end{eqnarray}
The complex $\dagger$-closed
algebra generated is $\dagger$-isomorphic to the complex $8\times 8$
matrices, so its irreducible representations on a Hilbert space all
can be identified with a Jordan-Wigner representation on $3$ two-level
(qubit) systems. Explicitly, one such representation maps 
\begin{eqnarray}
\gamma_{2m-1}^{\;} &=& \left (\prod_{\ell=1}^{m-1} \sigma_{z}^\ell \right )\sigma_{x}^{m} , \nonumber \\ 
\gamma_{2m}^{\;} &=&  \left (\prod_{\ell=1}^{m-1} \sigma_{z}^\ell \right )\sigma_{y}^{m} ,
\label{majoranadef}
\end{eqnarray}
where $\sigma^m_\tau$, $m=1,2,3$ and $\sigma_\tau=\sigma_x, \sigma_y, \sigma_z$, are Pauli matrices and 
 we have chosen a particular sign convention without physical 
consequences.

In the following, we confine ourselves to the $15$ physically 
measurable ``parity'' observables 
\begin{eqnarray} \notag
P_{jk}=i\gamma_{j}^{\;}\gamma_{k}^{\;} \ , \ 1\leq j < k\leq 6. 
\end{eqnarray}
The total parity $\cP=-i\prod_j\gamma_j$,
commutes with every other parity observable, 
partitions the full Hilbert space into even ($\cP=+1$) and odd ($\cP=-1$) parity subspaces. These subspaces are
invariant under the action of any parity operator and are isomorphic to logical two-qubit subspaces  as illustrated by the mapping of  Table \ref{Isomorphism: 2qubit-3qubit}.
We use $X$, $Y$, $Z$ to denote logical Pauli operators
acting on these two-qubit subspaces.
\begin{table}[htb]
\begin{center}
\begin{tabular}{| c | c | c |}
\hline
$P_{36}$ & $P_{25}$ & $P_{14}$ \\ \hline
$P_{12}$ & $P_{34}$ & $P_{56}$  \\ \hline
$P_{45}$ & $P_{16}$ & $P_{23}$ \\ \hline
\end{tabular}
$\Longleftrightarrow$
\begin{tabular}{| c | c | c |}
\hline
$\sigma_y^2\sigma_y^3$ & $-\sigma_x^1\sigma_z^2\sigma_x^3$ & $\sigma_y^1\sigma_y^2$ \\ \hline
$-\sigma_z^1$ & $-\sigma_z^2$ & $-\sigma_z^3$  \\ \hline
$-\sigma_x^2\sigma_x^3$ & $\sigma_y^1\sigma_z^2\sigma_y^3$ & $-\sigma_x^1\sigma_x^2$ \\ \hline
\end{tabular}

\hspace*{-1.7cm}\rotatebox[origin=c]{-45}{$\Longleftrightarrow$}
\hspace*{1.7cm}\rotatebox[origin=c]{45}{$\Longleftrightarrow$}
\vspace*{0.2cm}

\hspace*{-1.5cm}
\begin{tabular}{| c | c | c |}
\hline
$\cP ZZ$ & $\cP XX$ & $\cP YY$ \\ \hline
$ZI$ & $IX$ & $\cP ZX$  \\ \hline
$IZ$ & $XI$ & $\cP XZ$ \\ \hline 
\end{tabular} 
\end{center}
\caption{(Top) Peres-Mermin-like magic squares using Majorana fermion parity operators and related three physical qubits. Operators in the same row or column commute. (Bottom) The ``emergent"  operators realize Peres-Mermin magic squares in the
(four-dimensional) even-parity ($\cP=+1$) or odd-parity  ($\cP=-1$) subspaces.}
\label{Table: Magic square}
\end{table}

We say a set of fermion parity measurements are compatible if the corresponding parity operators are mutually commuting. 
There are exactly $15$ maximal sets of compatible 
measurements,  which are given by 
%the sequentially ordered sets $S_n$, $n=1,\ldots,15$,
\begin{eqnarray}
&&\{P_{36},P_{25},P_{14} \}\ ,\ \{P_{12},P_{34},P_{56} \} \ , \  \{P_{45},P_{16},P_{23}\}\ , \ \nonumber \\ &&\{P_{36}, P_{12},P_{45} \} \ ,\
\{P_{25}, P_{34},P_{16}\}\ , \   \{P_{14},P_{56}, P_{23} \} \ , \ \nonumber \\
&&\{P_{35},P_{16},P_{24} \}\ , \    \{P_{46},P_{25},P_{13} \}\ , \   
\{P_{12},P_{35},P_{46} \} \ , \ \nonumber \\ &&\{P_{56},P_{24},P_{13} \}\ , \  
\{P_{46},P_{15},P_{23} \}\ , \ \{P_{34},P_{26},P_{15} \} \ , \ \nonumber \\
&&\{P_{36},P_{24},P_{15}\} \ ,\ \{P_{13},P_{26},P_{45} \} \ , \ \{P_{35},P_{26},P_{14} \} \ . \ \nonumber \end{eqnarray}
We can select the first $6$ of those sets and form a $3\times3$ table which works like a Peres-Mermin  magic square \cite{Peres1990-1, Mermin1990} up to a unitary equivalence in both even and odd parity subspaces, as illustrated in Table \ref{Table: Magic square}.

\subsection{Quantum Self-testing}\label{subsection: self testing}

A \textit{self-testing} protocol aims to certify that both 
an unknown state $\ket{\Psi}$ and a set of unknown measurements
are equivalent to 
an ideal, usually entangled, state $|\hat\Psi\rangle$
and a set of ideal measurements.
Importantly, the certification does not rely on any 
assumptions about the state and measurements,
other than the assumption that certain pairs of measurement operators
commute.
The protocol involves
repeatedly performing different sets of pairwise commuting measurements.
If the ideal measurement statistics are obtained,
then the state and measurements are uniquely determined,
up to some notion of equivalence.
This was first observed by Popescu and 
Rohrlich~\cite{Popescu1992}, who proved that any
state that maximally violates a particular Bell inequality (the CHSH inequality)
is equivalent to a singlet state of two qubits. The equivalence is up to a
local isometry, because the measurement statistics are unaffected by a local change 
of basis and by the existence of an auxiliary
subsystem on which the measurements act trivially.
The notion of self-testing 
was formalized by Mayers and Yao~\cite{Mayers2004}, and since then,
self-testing protocols for many other states 
and measurement scenarios~\cite{McKague2011, Wu2014, Kaniewski2016, Coladangelo2017, Kalev2017, Breiner2019}
have been discovered. Such protocols are 
often called \textit{device-independent} because they rely only on the
statistics of measurement outcomes, and not on any physical assumptions
about the measurement apparatus.

Two important notions in the self-testing literature are 
that of \textit{rigidity} and \textit{robustness}.
A measurement scenario is rigid if achieving
the ideal expectation values uniquely determines the state and measurements,
up to a local isometry.
In any real experiment, however, the ideal statistics
will not be achieved exactly due to errors in the
state preparation and measurements.
Thus, any practical self-testing protocol
must include a robustness statement. Robustness implies that the
state and measurements are still determined approximately
if the statistics deviate from the ideal case by a small amount.
There are fewer known robustness results
for measurements than for states \cite{Kaniewski2017}.
Our main results are a rigidity theorem
and a robustness theorem for
Majorana fermion parity operators.

Our results differ from previous
self-testing results in a few respects. First, self-testing
scenarios typically involve two or more parties whose 
measurement operators commute due to a locality assumption.
The locality can be physically enforced, for example, by
requiring the measurements made by different parties to be
spacelike separated. In the scenario we consider, there is no
natural notion of locality. Therefore, we do not assume that
the full Hilbert space ${\cal H}$ factors as a tensor product. Nonetheless,
as we show, if the measurement operators have
the ideal expectation values, then there is a
natural tensor product decomposition. The unknown
state $\ket{\Psi}$ is maximally entangled with respect to this emergent
tensor product structure. Second, robust self-testing statements
are often formulated in terms of an extraction map,
which acts on a joint system comprised of the 
unknown Hilbert space and a reference
Hilbert space with a known dimension. In this formulation,
a robustness statement asserts that there exists
such an extraction map, such that the output state 
of the reference system has high fidelity with the ideal state~\cite{McKague2012a, McKague2012b}.
Our theorems avoid using an extraction map
and instead directly construct a four-dimensional
subspace of $\cH$.
In the rigid case, we show that the subspace contains $\ket{\Psi}$
and is invariant under the action of each of the measurement operators.
In the case of errors we define an ideal state and ideal operators
on the subspace and we lower bound the fidelities of the
actual state and measurement operators.

\section{Statement of Results}\label{sec:statements}

We consider an experimental setup ideally involving 6 Majorana modes and 15 parity
operators. However, we do not assume Majorana fermion parity operators from the outset 
as our aim is to infer Majorana behavior solely from the
statistics of measurement outcomes.
We assume that a quantum system is prepared in some unknown state $\rho$. Since
any mixed state has a pure state extension,
we can take $\rho=\ket{\Psi}\bra{\Psi}$ to be pure without loss of generality.
We also assume a set of unknown measurements,
each of which is given by a two-element positive operator-valued measure (POVM) $Q_r=\{Q_{r,0}, Q_{r,1}\}$.
Here $r\in\{1,\ldots,15\}$ labels the measurement and for all $r$, $Q_{r,0}+Q_{r,1}=\mathds 1$ and
$Q_{r,a}\ge0$, with $a=0,1$.
We assume that $[Q_{r,a}, Q_{s,b}]=0$ whenever $r$ and $s$ correspond
to parities having no Majorana modes in common.
%\MKc{As Karl pointed out, compatibility of POVMs is more general than ``the measurement operators commute''. We assume that the measurement operators commute. I haven't seen a proof that it is enough to assume that the POVMs are compatible. But the commutativity assumption is natural (if not fundamentally required) for us given the physics of the measurement.}
We emphasize that no other
assumptions about the state or measurements are made. In particular, we do not
assume the dimension of $\cH$ or that $\cH$ factors as a tensor product.

\begin{figure}[htb]
\centerline{\includegraphics[scale=0.28]{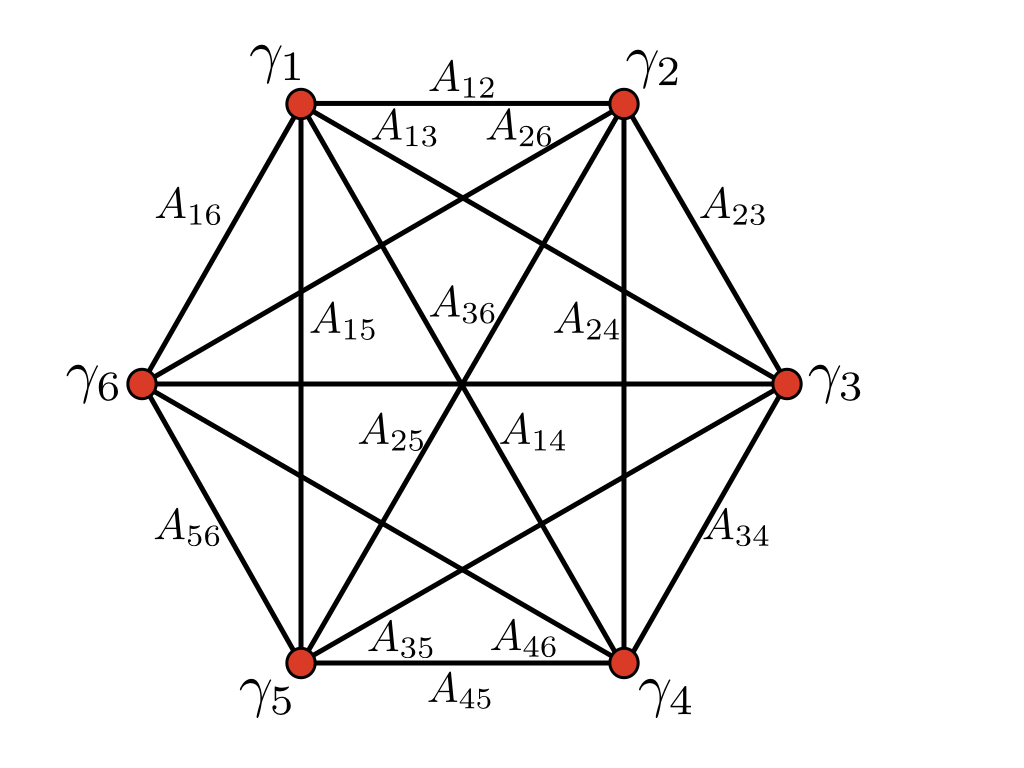}}
 \caption{Six Majorana fermion modes $\gamma_l$ indicated by 6 vertices. Each edge in the 
 $K_6$ graph represents a claimed Majorana fermion parity measurement $A_{jk}$ between $\gamma_l$ and $\gamma_k$.}
 \label{6-vertices}
\end{figure}
Before proving our main results, we first show that the POVM elements can be taken
to be orthogonal projectors without loss of generality.
In general, by Neumark's dilation theorem~\cite{Peres1990-2},
any POVM can be realized by a projective measurement on an extended Hilbert
space. In our case, we further 
require that the projectors obtained by dilation 
have the same pairwise commutativity structure
that was assumed of the POVMs. This is accomplished by the
following proposition, which we prove in Sec.~\ref{sec: Rigidity}.

\begin{proposition}\label{prop: dilation}
	Let $\{Q_{r,a}\}$ be a set of two-outcome POVM elements.
    Then there exist projectors $\hat{Q}_{r,a}$ satisfying
    $\Tr(\hat Q_{r,a}(\rho\otimes\ket{0}\bra{0}_r))=\Tr(Q_{r,a}\rho)$ for all density operators $\rho$,
    and $[\hat{Q}_{r,a},\hat{Q}_{s,b}]=0$ whenever
    $[Q_{r,a}, Q_{s,b}]=0$.
\end{proposition}

Having extended the POVMs to projective measurements,
we can define the Hermitian operators
$A_r = 2\hat{Q}_{r,0}-\mathds 1$. 
Note that $A_r^2=\mathds 1$, and so each $A_r$ is unitary
and has eigenvalues in $\{-1,+1\}$.
Such operators are called \textit{Hermitian involutions}.
These operators can 
be visualized as edges on $K_6$, the complete graph on 6 vertices,
as shown in Fig.~\ref{6-vertices}. The vertices
correspond to Majorana modes, and two operators commute if
their associated edges do not share a vertex. When convenient, we will
use a double index as in $A_{jk}$ to denote the operator associated 
with edge $(j,k)$. The maximal sets of commuting
observables are given by perfect matchings on $K_6$.

Our self-testing theorems apply to any
set of six parities corresponding to a cycle subgraph $G\subseteq K_6$. For concreteness, we take $G$ to be the cycle
whose edge set is $E=\{(1,2),(2,3),(3,4),(4,5),(5,6),(1,6)\}$.
We refer to a maximal set of commuting
parity operators in $G$ as a \textit{context}.
We arrange the six unknown operators into a 2-by-3 table
where the operators in each row and column
form a context (see Table \ref{Table unknown} (Left)). 

Any two operators $A_r$ and $A_s$ not in the same row or column correspond to
edges $r$ and $s$ that are adjacent in $G$, which we denote $r\sim s$.
The ideal fermionic parity operators corresponding to adjacent edges anti-commute.
Since 6 Majorana modes acting on a given parity (even or odd) sector 
encode 2 logical qubits,
the ideal operators can be any
set of 6 logical two-qubit Pauli operators with the ideal
commutation and anti-commutation relations.
For concreteness, we fix a basis in which the ideal operators are as in Table \ref{Table unknown} (Right).
\begin{table}[h] \label{Table unknown}
\begin{tabular}{| c | c | c |}
\hline
$A_{12}$ & $A_{34}$  &$A_{56}$ \\
 \hline
$A_{45}$  & $A_{16}$ & $A_{23}$ \\
\hline
\end{tabular}
\hspace*{0.5cm}
$\Longleftrightarrow$
\hspace*{0.5cm}
\begin{tabular}{| c | c | c |}
\hline
$ZI$ & $IX$  & $ZX$ \\
 \hline
$IZ$  & $XI$ & $XZ$ \\
\hline
\end{tabular}
\caption{(Left) The six unknown operators and five contexts. (Right) The six ``emergent'' logical two-qubit Pauli operators.}
\end{table}

Let $\cR_i$ and $\cC_i$ be the sets of edges
in the $i$th row and column, respectively.
The \textit{ideal expectations} in our self-testing protocol 
are the following expectation values of products of observables in each context:
\begin{align}\label{eq: ideal statistics}
    \bra{\Psi}\prod_{r\in \cR_i}A_r\ket{\Psi} &= 1\quad i\in\{1,2\},\notag\\
    \bra{\Psi}\prod_{r\in \cC_i}A_r\ket{\Psi} &=
    \begin{cases}
      1, & i\in\{1,2\} \\
      -1, & i=3.
    \end{cases}
\end{align}
The ideal expectations are achieved by the ideal
state $|\hat\Psi\rangle=\frac{1}{\sqrt{2}}(\ket{00}+\ket{11}$.
We remark that our particular definition of the ideal expectations
is a choice of convention. A similar rigidity result 
for a different ideal state follows from
any similar set of ideal expectations where an odd number
of contexts have an expectation value of $-1$.

Let $\cA$ be the algebra generated by $\{A_r:r\in G\}$,
and let $V\subseteq \cH$ be the subspace defined by 
$V=\mathrm{span}\{A\ket{\Psi}:A\in\cA\}$. Let $P$ be the projector
onto $V$, and let $\bA_r = P A_r P$.

\begin{theorem}[Rigidity of Majorana Parities]\label{thm 1}
If the ideal expectations are satisfied,
then V is a 4-dimensional subspace and $\{\bA_r, \bA_s\}=0$,
for all $r,s\in G$ with $r\sim s$. Furthermore, the state $\ket{\Psi}$ satisfies $\prod_{r\in \cC_i} A_r \ket{\Psi} =\ket{\Psi}$ for $i\in\{1,2\}$.
\end{theorem}
The proof of the theorem given in the next section. As a
consequence of Theorem~\ref{thm 1}, a basis for $V$ can be chosen
in which the operators on the (Left) in Table \ref{Table unknown} equal those on the 
(Right) of the same Table \ref{Table unknown}, and $\ket{\Psi}=|\hat\Psi\rangle$.

 In practice, experimental measurements do not satisfy the ideal expectations
 due to imperfections in the state preparation and measurements.
 We say that the ideal expectations are satisfied to within error $\epsilon$ if
\begin{align}\label{eq: Error statistics}
    \bra{\Psi}\prod_{r\in \cR_i} A_r\ket{\Psi}&\ge 1-\epsilon\quad i\in\{1,2\},\notag\\
    \pm\bra{\Psi}\prod_{r\in \cC_i} A_r\ket{\Psi}&\ge 1-\epsilon \quad i\in\{1,2,3\},
\end{align}
where the minus sign in the second line is used for the third column only.
In the presence of errors, the subspace $V$ is no longer invariant
under the action of the operators $A_r$. However, the protocol is
still robust in the following sense. There exists
an ideal subspace $\hat{V}$ of dimension 4, along with an ideal state
$|\hat{\Psi}\rangle$ and ideal operators $\hat{A}_r$ whose fidelities
with respect to the actual state and operators are close to 1,
within errors linear in $\epsilon$.
Here the state fidelity is $F(|\hat{\Psi}\rangle,\ket{\Psi})=|\langle\hat{\Psi}\ket{\Psi}|^2$,
and the operator fidelity is defined as $F(\hat A_r,A_r)=\frac{1}{4}\mathrm{Tr}(\hat A_r A_r)$.
Formally, we have the following:

\begin{theorem}[Protocol Robustness]\label{thm 2}
If the ideal expectations are satisfied within error $\epsilon$,
then there exists $\hat V\subseteq \cH$, with $\mbox{dim}(\hat V)=4$,
Hermitian involutions $\hat A_{r} : \hat V \rightarrow \hat V$ for each $r\in G$
such that $\{\hat A_r, \hat A_s\}=0$ if $r\sim s$ and $[\hat A_r, \hat A_s]=0$ otherwise,
and a state $\hat{\ket{\Psi}}\in \hat V$ such that $\prod_{r\in \cC_i} \hat A_r \hat{\ket{\Psi}} =\hat{\ket{\Psi}}$ for $i\in\{1,2\}$, and such that they satisfy
\begin{align*}
     F(|\hat{\Psi}\rangle,\ket{\Psi})&\ge 1- \epsilon_0,\\
    F(\hat A_{r}, A_r)&\ge 1-\epsilon_i,\quad\quad\, \mbox{for}\ r \in \cC_{i},\\
\end{align*}
where $\epsilon_0=14\epsilon$, $\epsilon_1=0$, $\epsilon_2=25\epsilon/2$, and $\epsilon_3=(\sqrt{2}+\sqrt{14}+\sqrt{44})^2\epsilon/2\simeq69.5\epsilon$.
\end{theorem}

The proof of the above theorem is given in Sec.~\ref{sec: Robustness},
with some details deferred to Appendix~\ref{Appendix: state fidelity} and~\ref{Appendix: equation to derive 3rd column}.
The perfect fidelity of the first column operators is due to a choice of basis.
For simplicity, we have chosen our error bounds to be equal. 
Our results can be generalized to the case
of unequal errors for different contexts, 
however, we do not carry out this analysis here.

\section{Rigidity of Majorana Fermion Parity Measurements}\label{sec: Rigidity}

In this section we first prove Proposition~\ref{prop: dilation},
and then we prove Theorem \ref{thm 1} with the help of some lemmas.

%\subsection{Stinespring Dilation Extension}
\subsection{From POVM to Projective Measurements}

Given any POVM $\{Q_{a}\}$, $a=0, \ldots, d-1$,
acting on system $\cS$, Neumark proved~\cite{Peres1990-1} that 
there exists a projective measurement $\{\hat Q_a\}$
on an extended system $\cS \otimes \cE$,
with $\Tr(\hat Q_a(\rho\otimes\ket{0}\bra{0}_{\cE}))=\Tr(Q_a\rho)$ 
for all density operators $\rho$ on $\cS$, where the dimension of 
the extension $\cE$ equals the number of elements $d$ in the POVM.
The projectors are of the form
\begin{equation}
\hat{Q}_a=U_{\cS \cE}^{\dag} (\mathds 1_\cS \otimes \ket{a}\bra{a}_\cE) U_{\cS \cE},
\label{Stinespring}
\end{equation}
where $U_{\cS \cE}^{\;}$ is a unitary on the extended Hilbert space.

In the case of multiple POVMs $Q_{r}$,
where $r$ labels the POVM with elements $\{Q_{r,a}\}$, the POVMs can be extended
to projectors $\{\hat{Q}_{r,a}\}$ by adding several ancillas ($\otimes_r \cE_r$), one for each POVM.
The situation is depicted in Fig.~\ref{fig:dilation}.
Our task is to prove that for
$d=2$, $[\hat{Q}_{r,a},\hat{Q}_{s,b}]=0$ whenever
    $[Q_{r,a}, Q_{s,b}]=0$.

\begin{figure}[htb]
\includegraphics[width=1.0 \columnwidth]{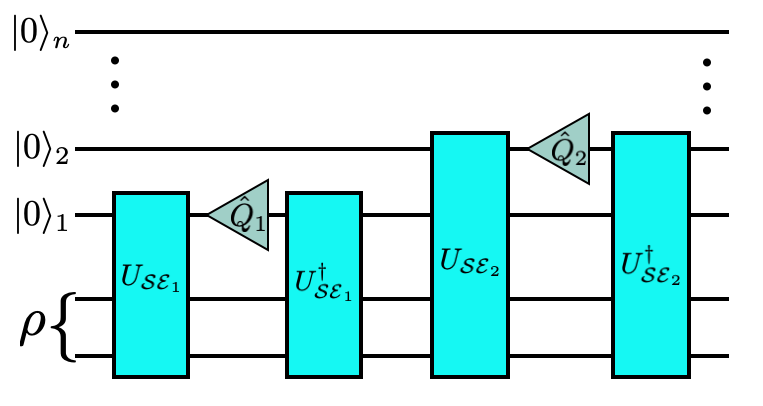}
 \caption{Dilation extension in quantum circuit language. }
 \label{fig:dilation}
\end{figure}

\begin{proof}[Proof of Proposition \ref{prop: dilation}]
Define the projector $\hat Q_{r,a}$  on the extended system $\cS \otimes \cE_r$ corresponding to POVM element $Q_{r,a}$ according to Eq.~\eqref{Stinespring}, with
\begin{eqnarray*}
U_{\cS \cE_r}\,\ket{\Psi}\otimes \ket{a}_r &=&\sum_{b=0}^1(-1)^{ab}\sqrt{Q_{r,a\oplus b}} \, \ket{\Psi}\otimes \ket{b}_r,
\end{eqnarray*}
where $\ket{a}_r$ is the basis-vector of the extension $\cE_r$ with $a\in\{0,1\}$ and $\oplus$ denotes addition mod 2.
We first show that $U_{\cS\cE_r}$ is unitary. Indeed, we have
\begin{multline*}
    \bra{\Phi}\bra{a}_r U_{\cS\cE_r}^{\dag}U_{\cS\cE_r}\ket{\Psi}\ket{b}_r \\
    = \sum_{c=0}^1(-1)^{(a+b)c}\bra{\Phi}\sqrt{Q_{r,a\oplus c}}\sqrt{Q_{r,b\oplus c}}\ket{\Psi}.
\end{multline*}
If $a=b$, then the last line above equals
\begin{equation*}
    \bra{\Phi}\sum_{c=0}^1 Q_{r,c}\ket{\Psi} = \bra{\Phi}\ket{\Psi},
\end{equation*}
whereas if $a\ne b$, it equals
\begin{equation*}
    \bra{\Phi}\big(\sqrt{Q_{r,a}}\sqrt{Q_{r,b}}-\sqrt{Q_{r,b}}\sqrt{Q_{r,a}}\,\big)\ket{\Psi} = 0,
\end{equation*}
where we've used $Q_{r,b} = \mathds1_{\cS}-Q_{r,a}$ and thus
$\sqrt{Q_{r,a}}$ commutes with $\sqrt{Q_{r,b}}$. Therefore, $U_{\cS\cE_r}$ is unitary.
Next, using the cyclic property of the trace, we compute
\begin{align*}
&\Tr(\hat{Q}_{r,a}(\rho\otimes \ket{0}\bra{0}_r))\\
=&\Tr( U_{\cS \cE_r}(\rho\otimes \ket{0}   \bra{0}_r)U_{\cS \cE_r}^{\dag}(\mathbb {1}_{\cS}\otimes \ket{a}  \bra{a}_r))\\
=&\Tr(\sqrt{Q_{r,a}}\rho \sqrt{Q_{r,a}})\\
=&\Tr(Q_{r,a}\rho).
\end{align*}
Finally, for $s\ne r$, in a block matrix representation with respect to the $\ket{a}_r$ basis,
\begin{equation} \notag
    U_{\cS \cE_r}=\left( \begin{array}{cc}
   \sqrt{Q_{r,0}} & \sqrt{Q_{r,1}} \\
   \sqrt{Q_{r,1}} & -\sqrt{Q_{r,0}} \\
  \end{array}  \right), \ \hat{Q}_{s,b}=\left( \begin{array}{cc}
   Q_{s,b} & 0 \\
   0 & Q_{s,b}\\
  \end{array}  \right).
\end{equation}
Therefore, $[Q_{r,a},Q_{s,b}]=0$ implies that  $[U_{\cS \cE_r}, \hat Q_{s,b}]=0$, and hence $[\hat Q_{r,a}, \hat Q_{s,b}]=0$.
\end{proof}

\subsection{Rigidity of the State and Observables}

We begin with a lemma that states
that operators with adjacent edges anti-commute in their action on $\ket{\Psi}$.
Since each $A_r$ has eigenvalues in $\{-1,+1\}$,
the ideal expectations are satisfied only if $\ket{\Psi}$ is
a $\pm1$ eigenstate of the products of operators in each context.
\begin{align}\notag
    \prod_{r\in \cR_i}A_r\ket{\Psi} &=\ket{\Psi}\quad i\in\{1,2\},\\
    \prod_{r\in \cC_i}A_r\ket{\Psi} &= \pm\ket{\Psi}\quad i\in\{1,2,3\},
    \notag
\end{align}
where again the minus sign in the last equation is for column $3$ only.
Using the fact that $A_r^2=\mathds1$, and the commutativity
of operators in each context, we can move operators freely between the
left and right sides of the above equations. For example, the identities
\begin{align*}
    A_{12}A_{34}\ket{\Psi} &= A_{56}\ket{\Psi},\\
    A_{34}\ket{\Psi} &= A_{16}\ket{\Psi},
\end{align*}
hold for row $1$ and column $2$, respectively.

\begin{lemma} \label{MS lemma}
Suppose the ideal expectations are satisfied. Then
$\{A_r,A_s\}\ket{\Psi}=0$, for $r\sim s$.
\end{lemma}

\begin{proof}
We show that $\{A_{12},A_{16}\}\ket{\Psi}=0$.
Making repeated use of identities such as the ones above,
we compute
\begin{multline*}
    A_{12}A_{16}\ket{\Psi}=A_{12}A_{34}\ket{\Psi}=A_{56}\ket{\Psi}\\
    =-A_{23}\ket{\Psi}=-A_{16}A_{12}\ket{\Psi}=-A_{16}A_{12}\ket{\Psi}.
\end{multline*}
By symmetry of the table,
a similar argument shows that the same
relation holds for any $A_r$ and $A_s$ with $r\sim s$.
\end{proof}

We now construct a subspace and show that it is 
invariant under the action of $\cA$. Define $V'\subseteq \cH$ by
\begin{equation} \notag
   V' = \mathrm{span}\{\ket{\Psi}, A_{12}\ket{\Psi}, A_{16}\ket{\Psi}, A_{12}A_{16}\ket{\Psi}\}.
\end{equation}
\begin{lemma}
$A_rV'\subseteq V'$ for all $r\in G$.
\end{lemma}
\begin{proof}
Since $A_{12}^2=\mathds1$, $A_{12}V'\subseteq V'$. To see that $A_{16}V'\subseteq V'$, note that Lemma \ref{MS lemma} implies $A_{16}A_{12}\ket{\Psi}=-A_{12}A_{16}\ket{\Psi}$ and $A_{16}A_{12} A_{16}\ket{\Psi}=-A_{12}\ket{\Psi}$.
We next check that $A_{34}V'\subseteq V'$.
This follows from $A_{34}\ket{\Psi}= A_{16}\ket{\Psi}$,
and from the fact that $A_{34}$ commutes with $A_{16}$ and $A_{12}$.
By symmetry of the table, we also have that $A_{45}V'\subseteq V'$.
It remains to show that $A_{56}$ and $A_{23}$ act invariantly on $V'$.
Explicity,
\begin{align*}
A_{56}\ket{\Psi}&=A_{12} A_{34}\ket{\Psi}\in V'\\
A_{56}A_{12} \ket{\Psi}&=A_{12}A_{56}\ket{\Psi}\in V' \\
A_{56}A_{16} \ket{\Psi}&=A_{56}A_{34}\ket{\Psi}=A_{12}\ket{\Psi}\in V'\\
A_{56}A_{12}A_{16} \ket{\Psi}&=A_{56}A_{12}A_{34}\ket{\Psi}=\ket{\Psi}\in V'.
\end{align*}
\ignore{should the second line be written as  $A_{56}A_{12} \ket{\Psi}=A_{12}A_{56}=A_{34}\ket{\Psi}$.\\}
Similarly, by symmetry, we also have $A_{23} V'\subseteq V'$.
\end{proof}

Having shown that $V'$ is an invariant
subspace,
it follows that $V'=V=\cA\ket{\Psi}$.
We can now work with the operators restricted onto $V$.
Let $\bA_r=PA_rP$, with $P$ the projector onto $V$.
Note that $AP=PAP$ for all $A\in\cA$.
The next Lemma states that commutativity and anti-commutativity of operators
on a full Hilbert space is preserved under restriction onto a subspace.
\begin{lemma} \label{Commutativity}
Let $A$ and $B$ be Hermitian involutions, and let $P$ be a
projector such that $AP=PAP$. Then
\begin{align*}
[A,B]=0 &\implies [PAP,PBP]=0,\\
\{A,B\}=0 &\implies \{PAP,PBP\}=0.
\end{align*}
\end{lemma}

\begin{proof}
$(PAP)(PBP)=PAPBP=(PAP)^{\dag}BP=(AP)^{\dag}BP=PABP=\pm PBAP=\pm PBPAP=\pm(PBP)(PAP)$,
with the plus or minus sign depending on whether $A$ and $B$ 
commute or anti-commute, respectively.
\end{proof}

We are now ready to prove Theorem $\ref{thm 1}$.
\begin{proof}[Proof of Theorem \ref{thm 1}]
We first determine the action on $V$ of the operators in $\cC_1$ and $\cC_2$. From Lemma \ref{Commutativity},
both $\bA_{12}$ and $\bA_{16}$ commute with both $\bA_{34}$ and $\bA_{45}$, and therefore
\begin{equation} \notag
    [\bA_{12}\bA_{16}, \bA_{45}\bA_{34}]=0.
\end{equation}
Since $\bA_{jk}^2=P$, $\bA_{12}\bA_{16}$
and $\bA_{45}\bA_{34}$ are unitary on $V$. Therefore, there exists an orthogonal basis for $V$ of simultaneous eigenstates of $\bA_{12}\bA_{16}$
and $\bA_{45}\bA_{34}$. Let $\ket{\alpha,\beta}$
be one such eigenstate satisfying
\begin{align*}
   \bA_{12}\bA_{16}\ket{\alpha,\beta}&=\alpha\ket{\alpha,\beta},\\
   \bA_{34}\bA_{45}\ket{\alpha,\beta}&=\beta\ket{\alpha,\beta},
\end{align*}
and also satisfying $\bra{\alpha,\beta}\ket{\Psi}\ne0$.
Such a state exists since $\ket{\Psi}\in V$.
From Lemma \ref{MS lemma}, $\{A_{12},A_{16}\}\ket{\Psi}=0$, 
and hence $\{\bA_{12}, \bA_{16}\}\ket{\Psi}=0$ by Lemma \ref{Commutativity}.
Similarly, $\{\bA_{34}, \bA_{45}\}\ket{\Psi}=0$.
These two equations imply $\alpha+\bar{\alpha}=0$ and $\beta+\bar{\beta}=0$, and thus $\alpha,\beta\in\{i,-i\}$, where $\bar \alpha$ denotes the complex conjugation of $\alpha$.

Now, define $\ket{\bar \alpha,{\beta}}=\bA_{12}\ket{\alpha, \beta}$. Note that
$\bA_{12}\bA_{16}(\bA_{12}\ket{\alpha,\beta})=\bA_{12}(\bA_{12}\bA_{16})^{\dag}\ket{\alpha,\beta}=\bar{\alpha}\,\bA_{12}\ket{\alpha,\beta}$,
and therefore
\begin{equation*}
\bA_{12}\bA_{16}\ket{\bar{\alpha},\beta}=\bar{\alpha}\ket{\bar\alpha,\beta}.
\end{equation*}
Similarly, defining $\ket{\alpha,\bar \beta} = \bA_{45}\ket{\alpha,\beta}$,
and $\ket{\bar\alpha,\bar \beta} = \bA_{12}\bA_{45}\ket{\alpha,\beta}$,
we see that $\ket{\alpha,\beta}$, $\ket{\alpha,\bar{\beta}}$, $\ket{\bar{\alpha},\beta}$, and $\ket{\bar{\alpha},\bar{\beta}}$
are joint eigenstates of $\bA_{12}\bA_{16}$ and $\bA_{45}\bA_{34}$ with
eigenvalues $(\alpha,\beta)$, $(\bar\alpha,\beta)$,
$(\alpha,\bar\beta)$, and $(\bar\alpha,\bar\beta)$, respectively.
These eigenstates are pairwise orthogonal, since $\alpha,\beta\in\{i,-i\}$.
Therefore, $V$ is a 4-dimensional subspace.
Next, note that
\begin{multline*}
\bA_{16}\bA_{12}\ket{\alpha,\beta}=(\bA_{12}\bA_{16})^{\dag}\ket{\alpha,\beta}
=\bar\alpha\ket{\alpha,\beta}\\
=-\alpha\ket{\alpha,\beta}=-\bA_{12}\bA_{16}\ket{\alpha,\beta}
\end{multline*}
Similar calculations applied to the remaining eigenstates of $\bA_{12}\bA_{16}$ and $\bA_{34}\bA_{45}$
show that $\{\bA_{12},\bA_{16}\}=0$ and
$\{\bA_{45},\bA_{34}\}=0$.
Therefore, there is a basis for $V$ in which
\begin{gather*}
\bA_{12}= ZI,\quad \bA_{34}= IX,\\
\bA_{45}= IZ,\quad \bA_{16}= XI.
\end{gather*}
We work in this basis for the remainder of the proof.
The next step is to determine $\ket{\Psi}$. From 
$\bA_{12}\bA_{45}\ket{\Psi}=ZZ\ket{\Psi}=\ket{\Psi}$, and
$\bA_{34}\bA_{16}\ket{\Psi}=XX\ket{\Psi}=\ket{\Psi}$,
it follows that
\begin{equation} \notag
    \ket{\Psi}=\frac{1}{\sqrt{2}}(\ket{0 0}+\ket{1 1}).
\end{equation}
The final step is to determine $\bA_{56}$ and $\bA_{23}$.
We begin with $\bA_{56}$. Since $[A_{56},A_{12}]=0=[A_{56},A_{34}]$, Lemma \ref{Commutativity} 
implies that $[\bA_{56},\bA_{12}]=0=[\bA_{56},\bA_{34}]$. Therefore, when 
expanded in a basis of two-qubit Pauli matrices, $\bA_{56}$ can only have non-zero weight on $II$, $ZI$, $IX$, and $ZX$.
However, $\bA_{56}\ket{\Psi}=\bA_{12}\bA_{34}\ket{\Psi}=ZX\ket{\Psi}$, 
which implies $\bA_{56}=ZX$, since the states $\ket{\Psi}$, $XI\ket{\Psi}$, 
$IX\ket{\Psi}$, and $ZX\ket{\Psi}$ are pairwise orthogonal. Similarly,
using $[\bA_{23},\bA_{45}]=0=[\bA_{23},\bA_{16}]$ and 
$\bA_{23}\ket{\Psi}=\bA_{45}\bA_{26}\ket{\Psi}$,
it follows that $\bA_{23}=XZ$.

\end{proof}

\section{Robustness to Errors}\label{sec: Robustness}
We now consider the situation where the ideal
statistics are satisfied to within error $\epsilon$.
We first prove an approximate version of Lemma ~\ref{MS lemma}.

\begin{lemma}\label{prop: Anticommutativity}
	Suppose the ideal expectations are satisfied to within error $\epsilon$. Then 	    $\|\{ A_r,  A_s\}\ket{\Psi}\|\le 5 \sqrt{2 \epsilon}$
	for all $r,s\in G$ with $r\sim s$.
\end{lemma}

\begin{proof}
We show that $\|\{A_{12},A_{16}\}\ket{\Psi}\|\le5\sqrt{2\epsilon}$.
For $r$ and $s$ in the same column,
\begin{align}
    \| A_{r}\ket{\Psi} \pm A_{s} {\ket{\Psi}}\| &= \sqrt{2(1\pm\bra{\Psi}A_rA_s\ket{\Psi})} \notag \\
    &\le \sqrt{2\big(1-(1-\epsilon)\big)}\notag\\
    &=\sqrt{2\epsilon}, \label{Difference: Column}
\end{align}
where in the first line, the plus sign is used for column 3, and the
minus sign for columns 1 and 2.
Similarly, for both rows of the table, with $r$, $s$, and $t$ in the same row, 
\begin{equation}
    \| A_{r}\ket{\Psi}- A_{s} A_t {\ket{\Psi}}\|\le \sqrt{2\epsilon}. \label{Difference: Row}
\end{equation}
Therefore, by a chain of triangle inequalities, and using the fact that
$\|U\ket{\Psi}\|=\|\ket{\Psi}\|$ for any unitary $U$,
\begin{multline}
\|(A_{12}A_{16}+A_{16}A_{12})\ket{\Psi}\| \le \|(A_{12}A_{16}-A_{12}A_{34})\ket{\Psi}\| \\
+ \|(A_{12}A_{34}-A_{56})\ket{\Psi}\| + \|(A_{56}+A_{23})\ket{\Psi}\| \\
+ \|(-A_{23}+A_{16}A_{45})\ket{\Psi}\| +\|(-A_{16}A_{45}+A_{16}A_{12})\ket{\Psi}\| \\
\le 5\sqrt{2\epsilon}.\notag
\end{multline}
Using a similar argument for any $r\sim s$, one can prove $\|\{A_r,A_s\}\ket{\Psi}\|\le 5 \sqrt{2 \epsilon}$.
\end{proof}

By a corollary to Jordan's lemma,
which we prove in Appendix~\ref{Appendix: Jordan},
$\cH$ decomposes as $\cH=\bigoplus_l\cH_l$, where each
$\cH_l$ is 4-dimensional and invariant under the action of
$A_{12}$, $A_{16}$, $A_{34}$, and $A_{45}$.
Since both of $A_{12},A_{16}$ commute with both of $A_{34},A_{45}$, 
each invariant subspace $\cH_l$
in the Jordan decomposition factors as a tensor product of two qubits.
Therefore, there is a basis for each subspace such that  
\begin{align}\label{eq: jordan angles}
    A_{12}&=\bigoplus_l ZI,\ A_{34}=\bigoplus_l (\cos \phi_l IX+\sin \phi_l IZ),\notag\\
    A_{45}&=\bigoplus_l IZ,\ A_{16}=\bigoplus_l (\cos \theta_l XI+\sin \theta_l ZI),
\end{align}
with $\theta_l,\phi_l\in[-\frac{\pi}{2},\frac{\pi}{2}]$.
Label this chosen basis for each $\cH_l$ as $\{\ket{00_l}, \ket{01_l}, \ket{10_l}, \ket{11_l}\}$.

Next, with respect to this Jordan decomposition, one can write $\ket{\Psi}$ as 
\begin{equation} \notag
    \ket{\Psi}=\sum_l\sqrt{p_l}\ket{\Psi_l},
\end{equation}
where each  $\ket{\Psi_l}\in\cH_l$ and $\sum p_l=1$. We define the ideal 
subspace $\hat V$ as the linear span
\begin{gather} \notag
    \hat V = \mathrm{span}\big\{\ket{ab}=\sum_l \sqrt{p_l}\ket{a b_l}: a,b \in \{0,1\}\big\}.
\end{gather}

We define the ideal operators to be logical
Pauli product operators in the above basis. Specifically, $\hat 
A_{12}= ZI$, $\hat A_{34}= IX$, $\hat A_{56}= ZX$, $\hat A_{45} = IZ$, 
$\hat A_{16}= XI$, $\hat A_{23}=XZ$. Finally, we define the ideal state with
respect to the above basis as
\begin{equation}\label{eq: ideal state}
    \hat{\ket{\Psi}}= \frac{\ket{00}+\ket{11}}{\sqrt 2}.
\end{equation}

\begin{proof}[Proof of Theorem \ref{thm 2}]
By definition, $\{\hat A_r, \hat A_s\}=0$ for $r,s\in G$ with $r\sim s$, 
and the ideal state satisfies $\prod_{r\in \cC_i} \hat A_r \hat{\ket{\Psi}} =\hat{\ket{\Psi}}$ for $i\in\{1,2\}$.

We first calculate the state fidelity.
Define $|\hat {\Psi}_l\rangle=\frac{1}{\sqrt{2}}(\ket{00_l}+\ket{11_l})$.
Using the freedom to choose the overall phase in each subspace,
we set $\langle \hat \Psi_l| \Psi_l\rangle\ge0$. Therefore,
\begin{equation} \notag
    \langle\hat \Psi\ket{\Psi} = \sum_l p_l\langle\hat \Psi_l\ket{\Psi_l} \ge
    \sum_l p_l|\langle\hat \Psi_l\ket{\Psi_l}|^2.
\end{equation}
It will be convenient to work in the $Y$ basis
$\{\ket{0_Y 0_Y}, \ket{0_Y 1_Y}, \ket{1_Y 0_Y}, \ket{1_Y 1_Y}\}$ within
each Jordan subspace $\cH_l$. Here $\ket{0_Y}=(\ket{0}+i\ket{1})/\sqrt{2}$, and $\ket{1_Y}=(\ket{0}-i\ket{1})/\sqrt{2}$.
We expand $\ket{\Psi_l}$  as
\begin{equation} \notag
     \ket{\Psi_l}=\sum_{a, b\in\{0,1\}} c_{ab}^l\ket{a_Y b_Y}, 
\end{equation}
where $\sum_{a,b}|c^l_{ab}|^2=1$. In this basis,
the ideal state is
\begin{equation} \notag
    |\hat \Psi\rangle = \sum_l \sqrt{\frac{p_l}{2}}\big(\ket{0_Y1_Y}+\ket{1_Y0_Y}\big).
\end{equation}

Therefore,
\begin{align}
    \langle\hat \Psi\ket{\Psi} &\ge
    \sum_l p_l|\langle\hat \Psi_l\ket{\Psi_l}|^2\notag\\
    &= \sum_l p_l\frac{1}{2}|c^l_{01}+c^l_{10}|^2\notag\\
    &= \sum_l p_l\big(|c^l_{01}|^2+|c^l_{10}|^2-\frac{1}{2}|c^l_{01}-c^l_{10}|^2\big)\notag\\
    &= 1-\sum_l p_l\big(|c^l_{00}|^2+|c^l_{11}|^2\big)-\sum_l p_l\frac{1}{2}|c^l_{01}-c^l_{10}|^2\notag
\end{align}
In Appendix~\ref{Appendix: state fidelity}, we prove that
\begin{align}
    \sum_l p_l\big(|c^l_{00}|^2+|c^l_{11}|^2\big)&\le\frac{13}{2}\epsilon,\label{eq: state fid proof 1}\\
    \sum_l p_l|c^l_{01}-c^l_{10}|^2 &\le\epsilon.\label{eq: state fid proof 2}
\end{align}
Thus, the state fidelity is bounded according to
\begin{align}\label{Fidelity: State}
     F(|\hat{\Psi}\rangle,\ket{\Psi})&\ge(1-7\epsilon)^2\notag\\
     &\ge 1 - 14\epsilon.
\end{align}

Next we bound the fidelity of the operators, starting with the first column.
From Eq.~\eqref{eq: jordan angles} and the definition of the ideal
operators, for $r\in\cC_1$, $\Tr(\hat{A}_r A_r)=4$ and so $F(\hat{A}_r, A_r)=1$.
For the second column, Eq.~\eqref{eq: jordan angles} implies
\begin{equation} \notag
    F(\hat{A}_{16}, A_{16}) = \sum_l p_l\cos{\theta_l}.
\end{equation}
Combining Lemma~\ref{prop: Anticommutativity} with Eq.~\eqref{eq: jordan angles},
\begin{align}\label{eq: Col 2 fid proof}
    \frac{25}{2}\epsilon &\ge \frac{1}{4}\|\{A_{12}, A_{16}\}\ket{\Psi}\|^2 \notag\\
    &= \sum_l p_l\sin^2{\theta_l} \notag\\
    &= 1 - \sum_l p_l\cos^2{\theta_l}\notag\\
    &\ge 1 - \sum_l p_l\cos{\theta_l}.
\end{align}
where the last line follows from
$\theta_l\in[-\frac{\pi}{2}, \frac{\pi}{2}]$ and therefore $\cos{\theta_l}\ge 0$.
Combining the last two equations yields
\begin{equation*}
    F(\hat{A}_{16}, A_{16})\ge1-\frac{25}{2}\epsilon.
\end{equation*}
A similar calculation shows that $F(\hat{A}_{23}, A_{23})\ge1-\frac{25}{2}\epsilon$ also.

The final step is to bound the fidelities of the third column operators.
We begin with $A_{56}$. The $\epsilon$ error in the first row implies
\begin{equation}\label{eq: Co 3 proof 1}
    \|A_{12}A_{34}\ket{\Psi}-A_{56}|{\Psi}\rangle\|\le \sqrt{2\epsilon},
\end{equation}
and from the state fidelity, Eq.~\eqref{Fidelity: State}, we get
\begin{equation}\label{eq: Co 3 proof 2}
    \|A_{56}(\ket{\Psi}-|\hat{\Psi}\rangle)\|=\|\ket{\Psi}-|\hat{\Psi}\rangle\|\le\sqrt{14\epsilon}.
\end{equation}
We show in Appendix~\ref{Appendix: equation to derive 3rd column} that
\begin{equation}\label{eq: Co 3 proof 3}
    \|\hat{A}_{34}|\hat{\Psi}\rangle - A_{34}\ket{\Psi}\|\le\sqrt{44\epsilon}.
\end{equation}
Applying the triangle inequality to Eqs.~\eqref{eq: Co 3 proof 1}-\eqref{eq: Co 3 proof 3},
and using $\hat A_{12} = A_{12}$,
\begin{equation} \notag
\|\hat {A}_{12}\hat{A}_{34}|\hat{\Psi}\rangle-A_{56}|\hat{\Psi}\rangle\|\le (\sqrt{2}+\sqrt{14}+\sqrt{44})\sqrt{\epsilon},
\end{equation}
which implies 
\begin{eqnarray} \notag
    \Re\langle\hat{\Psi}|\hat{A}_{12}\hat{A}_{34}A_{56}|\hat{\Psi}\rangle\ge 1-\epsilon_3.
\end{eqnarray}
Since $[A_{12},A_{56}]=0$ and $\hat A^2_{12}=P$, 
Lemma~\ref{Commutativity} implies $[\hat A_{12}, \bA_{56}]=0$. 
Therefore, when expanded in a basis of two-qubit Pauli matrices, 
$\bA_{56}$ can only have $I$ or $Z$ acting on the first qubit. 
Since $\hat A_{12}\hat A_{34}=ZX$, only the $ZX$ component of $\bA_{56}$ contributes to
$\Re \langle \hat \Psi |\hat{A}_{12}\hat{A}_{34}\bA_{56}|\hat{\Psi}\rangle$,
as any other component either gives zero or a purely imaginary number.
Hence,
\begin{align}
    \Re \langle \hat \Psi |\hat{A}_{12}\hat{A}_{34}\bA_{56}|\hat{\Psi}\rangle &=
    \Re \langle \hat \Psi |\hat{A}_{56}\bA_{56}|\hat{\Psi}\rangle \notag\\
    &= \frac{1}{4}\Tr(\hat A_{56}\bA_{56}) \notag\\
    &= F(\hat A_{56}, A_{56}), \notag
\end{align}
and it follows that $F(\hat{A}_{56}, A_{56}, )\ge 1- \epsilon_3$. By a similar argument one can show that $F(\hat{A}_{23}, A_{23})\ge 1- \epsilon_3$.
\end{proof}

\section{Discussion and Outlook}
\label{Discussion}

%\IA{Have not included yet: We have designed quantum self-testing protocols for measurement of Majorana fermion parities. Our protocols allow us to establish compatibility of experimental Majorana fermion parity data with detection of those non-local fermion modes.}

We have shown that measurements of Majorana fermion parities can be 
self-tested. This fact provides a powerful tool for determining how consistent 
experimental data is with the existence of Majorana fermion modes.
Experimentally, our protocol requires the ability to 
measure 6 observables $A_{jk}$, which ideally correspond to
parities $P_{jk}$ between consecutive Majorana modes,
as shown in Fig~\ref{fig: experimental set-up}.
The expectation value of any context (set of observables in a row or column of Table 
\ref{Table unknown} (Left)) can be 
obtained by repeatedly preparing a specific initial state, measuring the 
operators in that context, and using Eqs.~\eqref{Expectation: row} 
and \eqref{Expectation: column}.
Given data from the measurements in each context,
we can express the result by an operator $W$, often 
called a contextuality witness[cite] in the quantum information literature,
\begin{multline}
    W= A_{12}A_{34}A_{56}+A_{45}A_{16}A_{23}\\
    +A_{12}A_{45}+  A_{34}A_{16}- A_{56}A_{23}.   
\end{multline}
This witness obeys an 
inequality, $\langle W \rangle \le 3$,
for any classical assignments of outcomes to the measured observables,
and is upper-bounded in quantum theory by $\langle W \rangle \le 5$.
The witness is not unique, as different choices of the initial
state result in different combinations of signs of the terms in
$W$, according to Table~\ref{Table: Expectation values} in Appendix~\ref{Appendix: Ideal statistics}.

\begin{figure}[htb]
\centerline{\includegraphics[width=0.65\columnwidth]{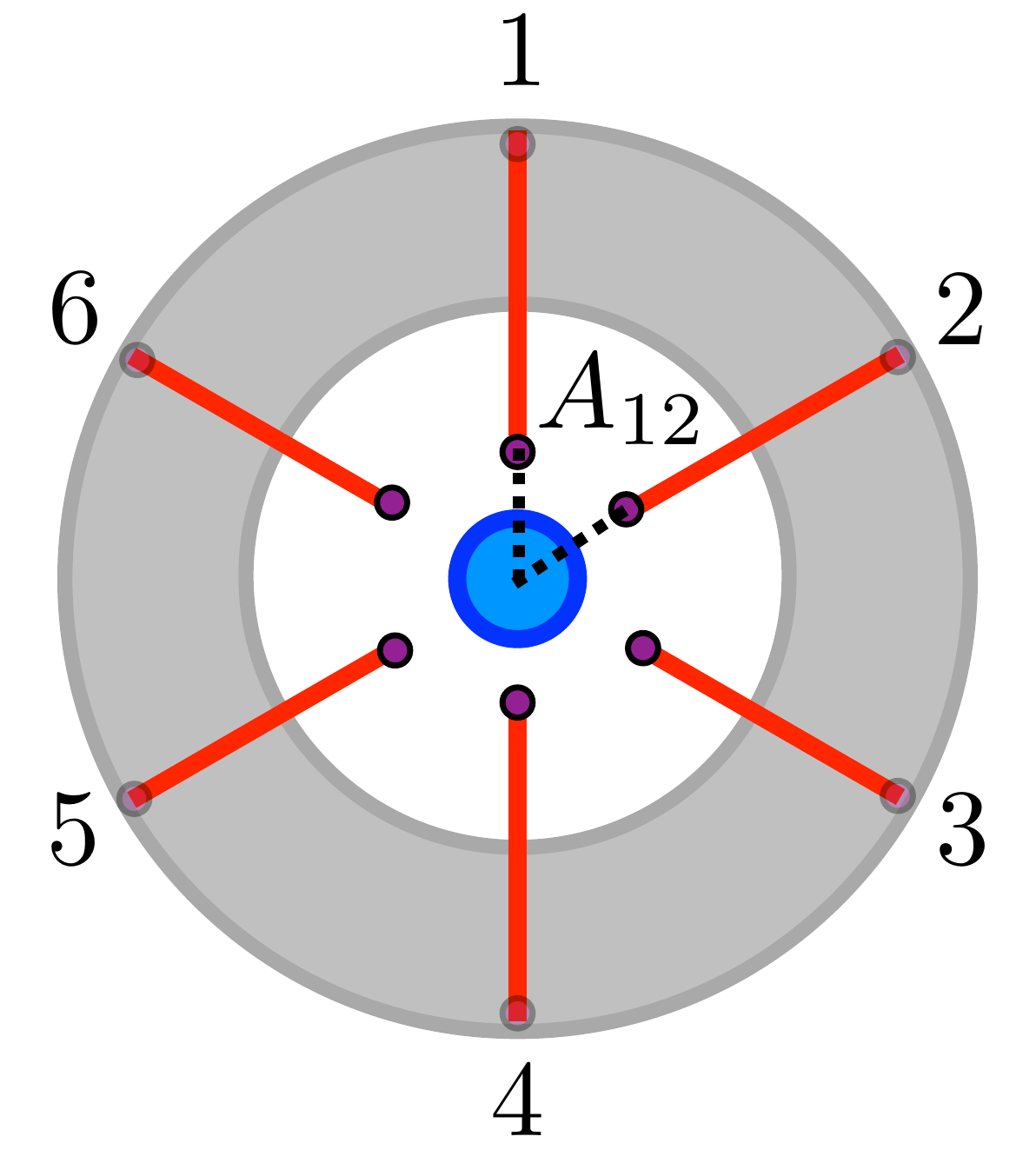}}
 \caption{Possible experimental setup to measure Majorana fermion parities $A_{jk}$ with eigenvalues $a_{jk}=\pm1$. Six (topological) semiconducting wires in proximity to a superconductor (indicated as a grey annulus) and a quantum dot in the center where parities are measured.}
 \label{fig: experimental set-up}
\end{figure}

For a fixed total parity (odd or even), $6$ 
Majorana fermions theoretically encode a two qubit subspace, where each qubit is encoded 
non-locally in $3$ Majorana modes in Fig. \ref{fig: experimental set-up}.
Our results, Theorems 1 and 2, imply that the maximal
value of $\langle W \rangle=5$ is obtained only if the
observables corresponding to adjacent parities anti-commute.
Furthermore, a small error
certifies that each $A_{jk}$ has high fidelity with the
corresponding Majorana parity operator. For example,
in a basis where the parities in the first row are perfect,
a $0.1\%$ error in the expectation value of each context implies
upper bounds on the error of the second and third column operators of $1.25\%$ and $6.95\%$, respectively.

We emphasize that although our proposal is to self-test
parities in Majorana modes, our protocol can be simulated
in other physical systems via the Jordan-Wigner mapping given
by Eq.~\eqref{majoranadef} and 
Table \ref{Table: Magic square}. Examples include trapped ions,
where Pauli product operators can be measured with
global entangling gates and the use of an ancilla~\cite{Leibfried2018}, or also
neutron beams entangled in energy, path and spin 
degrees of freedom~\cite{Cabello2008,Bartosik2009}.

We conclude with some open problems and suggestions
for future work. The robustness bounds in our Theorem 2
are certainly not tight, which raises the question
of how much they can be improved. It might be possible
to obtain a stronger robustness
statement using different methods such as those based
on a semidefinite programming hierarchy~\cite{Yang2014,Bancal2015},
or linear operator inequalities~\cite{Kaniewski2016}.
How does our formulation of robustness by constructing
an ideal subspace relate to the notion of robustness
as measured by an extraction map?
Finally, while our protocol certifies
a single state and a set of measurements, gates in
topological quantum computing are implemented by braiding.
It is therefore desirable to extend the protocol to include a self-test of braiding operations.

We strongly believe that certification of quantum
measurements in various physical scenarios is a promising technique for precision
measurement and quantum validation. We anticipate generalizations of our current approach to many other 
situations of physical interest exploring the frontiers of quantum mechanics. 

\begin{acknowledgements}
This research was conducted in part under the PREP
program with financial assistance from U.S. Department of Commerce,
National Institute of Standards and Technology.
Contributions to this article by workers
at the National Institute of Standards and Technology,
an agency of the U.S. Government, are not subject to U.S. copyright.

\end{acknowledgements}

\bibliography{library}

\newpage
\appendix

\section{Ideal Expectations and Statistics}\label{Appendix: Ideal statistics}

\begin{table}[htb]
    \begin{center}
    (a)
    \begin{tabular}{|c|c|c|c|c|} \hline
      \multirow{2}{*}{{\sf Context}}  & \multicolumn{4}{c|}{ $\mbox{Pr}(a_r, a_s, a_t|A_{r}, A_{s}, A_t)$ }\\ \cline{2-5}
     & $+++$ & $+--$ & $-+-$ & $--+$ \\ \hline
    ${\cal R}_1$ & $0.25$ &$0.25$ &$0.25$ &$0.25$ \\ \hline
    ${\cal R}_2$ & $0.25$ &$0.25$ &$0.25$ &$0.25$ \\ \hline
    \end{tabular}
    \end{center}
    \begin{center}
    (b)
    \begin{tabular}{|c|c|c|c|c|} \hline
    \multirow{2}{*}{{\sf Context}} & \multicolumn{4}{c|}{ $\mbox{Pr}(a_r, a_s|A_{r}, A_{s})$ }\\ \cline{2-5}
       & $++$ & $+-$ & $-+$ & $--$ \\ \hline
    ${\cal C}_1$ & $0.5$ &$0$ &$0$ &$0.5$ \\ \hline
    ${\cal C}_2$ & $0.5$ &$0$ &$0$ &$0.5$ \\ \hline
    ${\cal C}_3$ & $0$ &$0.5$ &$0.5$ &$0$ \\ \hline
    \end{tabular}
    \end{center}
    \caption{Probability Distribution of outcomes for (a) rows ${\cal R}_i$,  and (b) columns ${\cal C}_i$, when the initial state is $\ket{a_{36}^{(0)}=+1, a_{25}^{(0)}=+1, a_{14}^{(0)}=-1}$. The probability of measuring ($A_r, A_s, A_t$) (or $(A_r, A_s)$) and obtaining ($a_r, a_s, a_t$) (or $(a_r, a_s)$), with $a_r, a_s, a_t \in \{-1,+1 \}$,  is denoted as $\mbox{Pr}(a_r, a_s, a_t|A_{r}, A_{s}, A_t)$ (or $\mbox{Pr}(a_r, a_s|A_{r}, A_{s})$). Since the total parity of the state is $\cP=+1$, it has no amplitude in states with odd total parity.}
    \label{Table: Probality Distribution}
\end{table}

Operators in each context of Table \ref{Table: Magic square} realize complete sets of commuting observables (CSCO), allowing  preparation of states with definite parity assignments. In the following discussion, we consider the experimental preparation of an initial state with a well-defined  parity  $a_{36}^{(0)}$, $a_{25}^{(0)}$ and $a_{14}^{(0)}$, $a_{jk}^{(0)}\in \{-1,+1\}$, corresponding to the claimed parity observable $A_{jk}$. We then 
measure the contexts (sets of operators in  a row or column) of Table  \ref{Table unknown}. As an illustration, the ideal probability distribution of measurement outcomes of all contexts for the initial state $a_{36}^{(0)}=+1, a_{25}^{(0)}=+1, a_{14}^{(0)}=-1$ is given in Table  \ref{Table: Probality Distribution}. From the statistics of measurement outcomes one can calculate the expectation value of the product of the operators in each context by using
\begin{align}
     \langle A_{r}A_{s} \rangle&= \frac{\sum_{a_r, a_{s} } a_r a_{s} N(a_r,a_{s})}{\sum_{a_r, a_{s} }  N(a_r,a_{s})},\label{Expectation: row}\\
      \langle A_{r}A_{s} A_t \rangle&= \frac{\sum_{a_r, a_{s},a_t } a_r a_{s} a_t N(a_r,a_{s}, a_t)}{\sum_{a_r, a_{s}, a_t }  N(a_r,a_{s}, a_t)}, \label{Expectation: column}
\end{align}
where $N(a_r, a_s)$ (or $N(a_r, a_s, a_t)$) refers to the number of experimental outcomes with value $(a_r, a_s)$ (or $(a_r, a_s, a_t)$). Ideal expectation values of all contexts for different initial states with all possible $a_{36}^{(0)}, a_{25}^{(0)}, a_{14}^{(0)} \in \{-1,+1\}$ is given in Table \ref{Table: Expectation values}.

\begin{table}[htb]
    \centering
    \begin{tabular}{|c|c|c|c|c|c|}
    \hline
    {\sf Initial state}  & \multicolumn{5}{c|}{{\sf Expectation Values}}\\ \cline{2-6}
     $\ket{a_{36}^{(0)}, a_{25}^{(0)}, a_{14}^{(0)}}$ &  $\cR_1$ & $\cR_2$ & $\cC_1$ & $\cC_2$ & $\cC_3$ \\ \hline
    $\ket{+,+,+}$ & $-1$ & $-1$ & $-1$ & $-1$ & $-1$\\ \hline
    $\ket{+,+,-}$ & $+1$ & $+1$ & $+1$ & $+1$ & $-1$\\ \hline
    $\ket{+,-,+}$ & $+1$ & $+1$ & $+1$ & $-1$ & $+1$\\ \hline
    $\ket{+,-,-}$ & $-1$ & $-1$ & $-1$ & $+1$ & $+1$\\ \hline
    $\ket{-,+,+}$ & $+1$ & $+1$ & $-1$ & $+1$ & $+1$\\ \hline
    $\ket{-,+,-}$ & $-1$ & $-1$ & $+1$ & $-1$ & $+1$\\ \hline
    $\ket{-,-,+}$ & $-1$ & $-1$ & $+1$ & $+1$ & $-1$\\ \hline
    $\ket{-,-,-}$ & $+1$ & $+1$ & $-1$ & $-1$ & $-1$\\ \hline
    \end{tabular}
    \caption{Expectation values of product of operators in rows ($\cR_i$) and columns ($\cC_i$) for different initial states. }
    \label{Table: Expectation values}
\end{table}

\section{Jordan's Lemma} \label{Appendix: Jordan}
We prove a corollary, which we used in the main text,
of what is known as Jordan's Lemma in the quantum information literature.
A particularly simple proof of Jordan's Lemma
appears in Ref. [\onlinecite{Pironio2009}],
which we also include here for completeness.

\begin{lemma}[Jordan's Lemma]\label{lemma: Jordan}
    Let $A_1$ and $A_2$ be Hermitian involutions on a Hilbert space $\cH$.
    Then $\cH$ decomposes as a direct sum $\cH=\bigoplus_l\cH_l$,
    with $\dim \cH_l\in\{1,2\}$,
    and $A_1$ and $A_2$ act invariantly on each $\cH_l$.
\end{lemma}
\begin{proof}
$A_1A_2$ is unitary since $(A_1A_2)(A_1A_2)^{\dag}=A_1A_2A_2A_1=\mathds 1$.
Since $A_1A_2$ is unitary, there exists an orthonormal
basis for $\cH$ of eigenstates of $A_1A_2$.
Let $\ket{\alpha}$ be any such eigenstate, where
$A_1A_2\ket{\alpha}=\alpha\ket{\alpha}$.
Define $\ket{\bar\alpha}=A_1\ket{\alpha}$.
Then $A_1A_2\ket{\bar\alpha}=\bar\alpha\ket{\bar\alpha}$, since
\begin{equation*}
    A_1A_2A_1\ket{\alpha}=A_1(A_1A_2)^{\dag}\ket{\alpha}=\bar\alpha A_1\ket{\alpha}.
\end{equation*}
The span of $\ket{\alpha}$ and $\ket{\bar\alpha}$
is invariant under $A_1$, and also under $A_2$, since
\begin{align*}
    A_2\ket{\alpha}&=A_2A_1\ket{\bar\alpha}=(A_1A_2)^{\dag}\ket{\bar\alpha}=\alpha\ket{\bar\alpha},\\
    A_2\ket{\bar\alpha}&=A_2A_1\ket{\alpha}=(A_1A_2)^{\dag}\ket{\alpha}=\bar\alpha\ket{\alpha}.
\end{align*}
Thus any eigenstate of $A_1A_2$ defines an invariant
subspace of dimension at most 2, and these 
eigenstates span $\cH$, which completes the proof.
\end{proof}

\begin{corollary}
For $k,k' \in \{1,2\}$, let $A_k$ and $B_{k'}$ be Hermitian involutions
with $[A_k, B_{k'}]=0$. Then $\cH$ decomposes
as 
$\cH=\bigoplus_l (\cH_{a_l}\otimes\cH_{b_l})$,
with $\cH_{a_l}$ and $\cH_{b_l}$ of dimension
at most $2$, and
$A_k=\bigoplus_l (A_{k_l}\otimes\mathds 1_l)$
and $B_{k'}=\bigoplus_l (\mathds 1_l\otimes B_{k'_l})$.
\end{corollary}
\begin{proof}
Note that $[A_1A_2, B_1B_2]=0$. Thus, there exists an orthonormal
basis for $\cH$ of simultaneous eigenstates of $A_1A_2$ and $B_1B_2$. Let $\ket{\alpha,\beta}$
be any such eigenstate, where $A_1A_2\ket{\alpha,\beta}=\alpha\ket{\alpha,\beta}$
and $B_1B_2\ket{\alpha,\beta}=\beta\ket{\alpha,\beta}$.
Define $\ket{\bar\alpha,\beta}=A_1\ket{\alpha,\beta}$,
$\ket{\alpha,\bar\beta}=B_1\ket{\alpha,\beta}$, and
$\ket{\bar\alpha,\bar\beta}=A_1B_1\ket{\alpha,\beta}$.
Then $\mathrm{span}\{\ket{\alpha,\beta}, \ket{\bar\alpha,\beta}, \ket{\alpha,\bar\beta}, \ket{\bar\alpha,\bar\beta}\}$
maps isomophically to $\mathrm{span}\{\ket{\alpha}, \ket{\bar\alpha}\}\otimes\mathrm{span}\{\ket{\beta}, \ket{\bar\beta}\}$.
By the argument in the proof of Lemma~\ref{lemma: Jordan},
$A_k$ and $B_{k'}$ act invariantly on the first and second tensor factors,
and trivially on the second and first tensor factors, respectively.
\end{proof}

We remark that in the main text we assume that each Jordan
subspace $\cH_l$ has dimension 4. This is done without loss of
generality, since we can extend any smaller dimensional
subspace to $4$ dimensions, with all operators acting trivially on the extension.

%\begin{lemma}\label{Jordan}
%Let $A_k$ and $B_{k'}$ be Hermitian involutions,
%with $k,k' \in \{1,2\}$ and $[A_k, B_{k'}]=0$,
%acting on full Hilbert space $\cH$ of finite or countable infinite dimension.
%Then $\cH$ decomposes as a direct sum $\cH=\otimes_l (\cH_{a_l}\otimes\cH_{b_l})$, with $\dim \cH_{a_l},\ \dim \cH_{b_l}\in \{1,2\}$,
%such that $A_k$ acts on $\cH_{la}$ and $B_{k'}$ acts on $\cH_{lb}$.
%\end{lemma}
%\begin{proof}
%By definition unitary operators $A_1 A_2$ and $B_1B_2$ are compatible, therefore, there exist an orthogonal eigenbasis $\{\ket{\alpha \beta}\}$ such that  $A_1A_2\ket{\alpha,\beta}=\alpha\ket{\alpha,\beta}$ and  $B_1 B_2 \ket{\alpha,\beta}=\beta\ket{\alpha,\beta}$ with $|\alpha|=1=|\beta|$. Let  $\ket{\alpha_l,\beta_l}$ be one such eigenstate. We can say, $A_1 \ket{\alpha_l,\beta_l}= \ket{\bar \alpha_l,\beta_l}$, because 
%\begin{equation} \notag
 %   A_1 A_2\left(A_1 \ket{\alpha_l,\beta_l}\right)=A_1\left(A_1 A_2 \right)^\dagger \ket{\alpha_l,\beta_l}=\bar \alpha_l A_1 \ket{\alpha_l,\beta_l},
%\end{equation}
%which implies $A_2 \ket{\alpha_l,\beta_l}= \alpha_l \ket{\bar \alpha_l,\beta_l}$. Following similar argument one can show that $B_1 \ket{\alpha_l,\beta_l}= \ket{\alpha_l, \bar \beta_l}$ and $B_2 \ket{\alpha_l,\beta_l}= \beta_l \ket{\alpha_l, \bar \beta_l}$. By isomorphism one can map $\ket{\alpha_l,\beta_l}\mapsto\ket{\alpha_l}\otimes \ket{\beta_l}$ and can define $\cH_{a_l}=\mbox{span}\{\ket{\alpha_l},\ket{\bar \alpha_l}\}$,  $\cH_{b_l}=\mbox{span}\{\ket{\beta_l},\ket{\bar \beta_l}\}$, which completes the proof.
%\end{proof}

\section{Derivation of Eqs.~\eqref{eq: state fid proof 1} and \eqref{eq: state fid proof 2}.} \label{Appendix: state fidelity}
It will be convenient to work in the $Y$ basis
$\{\ket{0_Y 0_Y}, \ket{0_Y 1_Y}, \ket{1_Y 0_Y}, \ket{1_Y 1_Y}\}$ within
each Jordan subspace $\cH_l$. Here $\ket{0_Y}=(\ket{0}+i\ket{1})/\sqrt{2}$, and $\ket{1_Y}=(\ket{0}-i\ket{1})/\sqrt{2}$.
We expand $\ket{\Psi_l}$  as
\begin{equation} \notag
     \ket{\Psi_l}=\sum_{a, b\in\{0,1\}} c_{ab}^l\ket{a_Y b_Y}, 
\end{equation}
where $\sum_{a,b}|c^l_{ab}|^2=1$. In this basis,
the ideal state is
\begin{equation} \notag
    |\hat \Psi\rangle = \sum_l \sqrt{\frac{p_l}{2}}\big(\ket{0_Y1_Y}+\ket{1_Y0_Y}\big)
\end{equation}
Our first step is to bound $\sum_{l} p_l ( |c_{00}^l|^2+  |c_{11}^l|^2 )$. 
Applying Eq. \eqref{Difference: Column} to $\cC_1$ and $\cC_2$,
respectively, we obtain
\begin{equation} \notag
    \sum_l p_l (|c^l_{00}-c^l_{11}|^2 + |c^l_{01}-c^l_{10}|^2) \le \epsilon,
\end{equation}
\begin{equation} \notag
    \sum_l p_l (|e^{-i\phi_l}c^l_{00}+e^{i\theta_l}c^l_{11}|^2 + |e^{i\phi_l}c^l_{01}-e^{i\theta_l}c^l_{10}|^2) \le \epsilon,
\end{equation}
Eq.~\eqref{eq: state fid proof 2} follows from the first of these equations.
It also follows that
\begin{align}
    \sum_l p_l |c^l_{00}-c^l_{11}|^2 &\le \epsilon\label{eq: app coeff bound 1}.\\
    \sum_l p_l |e^{-i\phi_l}c^l_{00}-e^{i\theta_l}c^l_{11}|^2&\le \epsilon.\label{eq: app coeff bound 2}
\end{align}
Now,
\begin{align}
    &\abs{c^l_{11}}^2\abs{e^{i\theta_l}+e^{-i\phi_l}}^2\notag\\
    &= \abs{(e^{i\theta_l}c^l_{11}+e^{-i\phi_l}c^l_{00})-e^{-i\phi_l}(c^l_{00}-c^l_{11})}^2\notag\\
    &\le 2\big(\abs{e^{i\theta_l}c^l_{11}+e^{-i\phi_l}c^l_{00}}^2 + \abs{c^l_{00}-c^l_{11}}^2\big), \notag
\end{align}
\ignore{I think, in the second line
$e^{-i\phi_l}(c^l_{00}-c^l_{11})$ is wrong,
$e^{-i\phi_l}(c^l_{11}-c^l_{00})$ should be correct.\\}
where in the last line we used the fact that
$|x+y|^2\le2(|x|^2+|y|^2)$ for any $x,y\in\mathbb{C}$.
From Eqs.~\eqref{eq: app coeff bound 1} and  \eqref{eq: app coeff bound 2},
\begin{align}\label{eq: ap eq1}
    &\sum_l p_l \abs{c^l_{11}}^2\abs{e^{i\theta_l}+e^{-i\phi_l}}^2\notag\\
    &= \sum_l p_l \abs{c^l_{11}}^2(2+2 \cos(\theta_l+\phi_l))\le 4\epsilon
\end{align}
Similarly,
\begin{align}
    &\abs{c^l_{00}}^2\abs{e^{i\theta_l}+e^{-i\phi_l}}^2\notag\\
    &= \abs{(e^{i\theta_l}c^l_{11}+e^{-i\phi_l}c^l_{00})+e^{i\theta_l}(c^l_{00}-c^l_{11})}^2\notag\\
    &\le 2\big(\abs{e^{i\theta_l}c^l_{11}+e^{-i\phi_l}c^l_{00}}^2 + \abs{c^l_{00}-c^l_{11}}^2\big), \notag
\end{align}
from which it follows that
\begin{align}\label{eq: ap eq2}
    &\sum_l p_l \abs{c^l_{00}}^2\abs{e^{i\theta_l}+e^{-i\phi_l}}^2\notag\\
    &= \sum_l p_l \abs{c^l_{00}}^2(2+2 \cos(\theta_l+\phi_l))\le 4\epsilon.
\end{align}
Adding Eqs.~\eqref{eq: ap eq1} and \eqref{eq: ap eq2},
\begin{equation}\label{eq: cos(th + ph)}
    \sum_l p_l(|c^l_{00}|^2+|c^l_{11}|^2)(1+\cos(\theta_l+\phi_l))\le 4\epsilon.
\end{equation}
We now need the following lemma, which is similar to Lemma~\ref{prop: Anticommutativity}.
\begin{lemma}\label{lemma: deg 4 norm bound}
Suppose the ideal expectations are satisfied to within error $\epsilon$. Then 	    $\|A_{12}A_{34}\ket{\Psi}+A_{16}A_{45}\ket{\Psi}\|\le 3 \sqrt{2 \epsilon}$.
\end{lemma}
\begin{proof}
\begin{align}
&\|A_{12}A_{34}\ket{\Psi}+A_{16}A_{45}\ket{\Psi}\|\notag\\
&\le \|(A_{12}A_{34}-A_{56})\ket{\Psi}\| + \|(A_{56}+A_{23})\ket{\Psi}\|\notag\\
&+ \|(-A_{23}+A_{16}A_{45})\ket{\Psi}\|\notag\\
&\le 3\sqrt{2\epsilon} \notag
\end{align}
where the last inequality follows from Eqs.~\eqref{Difference: Column} and \eqref{Difference: Row}.
\end{proof}

From the result of the lemma, we obtain
\begin{multline} \notag
    \sum_l p_l\big( (|c^l_{00}|^2+|c^l_{11}|^2)|e^{i\theta_l}+e^{i\phi_l}|^2\\
    + (|c^l_{01}|^2+|c^l_{10}|^2)|e^{i\theta_l}-e^{-i\phi_l}|^2 \big) \le 18\epsilon,
\end{multline}
and therefore,
\begin{equation} \notag
    \sum_l p_l (|c^l_{00}|^2+|c^l_{11}|^2)|e^{i\theta_l}+e^{i\phi_l}|^2 \le 18\epsilon,
\end{equation}
or,
\begin{equation}\label{eq: cos(th - ph)}
    \sum_l p_l(|c^l_{00}|^2+|c^l_{11}|^2)(1+\cos(\theta_l-\phi_l))\le 9\epsilon.
\end{equation}
Adding Eqs.~\eqref{eq: cos(th + ph)} and \eqref{eq: cos(th - ph)},
\begin{align}
    \frac{13}{2}\epsilon &\ge \sum_l p_l(|c^l_{00}|^2+|c^l_{11}|^2)(1+\cos{\theta_l}\cos{\phi_l})\notag\\
    &\ge \sum_l p_l(|c^l_{00}|^2+|c^l_{11}|^2). \notag
\end{align}
where in the last inequality we used $\theta_l$,$\phi_l\in[-\frac{\pi}{2}, \frac{\pi}{2}]$,
and so $\cos{\theta_l}\ge 0$ and $\cos{\phi_l}\ge 0$.
This last inequality is Eq.~\eqref{eq: state fid proof 1} in the main text.

\section{Derivation of Eq.~\eqref{eq: Co 3 proof 3}}\label{Appendix: equation to derive 3rd column}

By definition,
\begin{multline}
    \langle \hat \Psi_l|(\mathds{1}-\hat{A}_{34} A_{34}) |\Psi_l\rangle= \frac{(c^l_{01}+c^l_{10})-(c^l_{01}e^{i\phi_l}+c^l_{10}e^{-i\phi_l})}{\sqrt{2}}. \notag
\end{multline}
Therefore,
\begin{multline}
    \Re \langle \hat \Psi_l|(\mathds{1}-\hat{A}_{34} A_{34}) |\Psi_l\rangle\\
    = \frac{(1-\cos \phi_l)\Re (c^l_{01}+c^l_{10})+\sin \phi_l \Im (c^l_{01}-c^l_{10})}{\sqrt{2}}\\
    \le(1-\cos \phi_l) +\frac{\sin \phi_l\Im(c^l_{01}-c^l_{10})}{\sqrt{2}}. \notag
\end{multline}

Summing over the Jordan subspaces, we get
\begin{align}
    &\sum_l p_l \Re \langle \hat \Psi_l|(\mathds{1}-\hat{A}_{34} A_{34}) |\Psi_l\rangle\notag\\
    &\le 1-\sum_l p_l \cos \phi_l+\frac{1}{\sqrt{2}}\sum_l p_l \sin \phi_l\Im(c^l_{01}-c^l_{10})\notag\\
    &\le 1-\sum_l p_l \cos^2 \phi_l+\frac{1}{\sqrt{2}}\sum_l p_l \sin \phi_l|c^l_{01}-c^l_{10}|\notag\\
    &= \sum_l p_l \sin^2 \phi_l+\frac{1}{\sqrt{2}}\sum_l p_l \sin \phi_l|c^l_{01}-c^l_{10}|. \label{Real part: 1-A34A34}
\end{align}
The second term in the last line is bounded by
\begin{align*}
    \sum_l p_l \sin \phi_l|c^l_{01}-c^l_{10}|&\le \sqrt{\sum_l p_l \sin^2 \phi_l}\sqrt{\sum_l p_l |c^l_{01}-c^l_{10}|^2}\\
    &\le \sqrt{\frac{25\epsilon}{2}}\,\sqrt{\epsilon}, \notag
\end{align*}
where we've used Eq.~\eqref{eq: state fid proof 2} and the same argument
leading to Eq.~\eqref{eq: Col 2 fid proof}, applied to $\phi_l$.
Therefore,
\begin{align*}
    \sum_l p_l \Re \langle \hat \Psi_l|(\mathds{1}-\hat{A}_{34} A_{34}) |\Psi_l\rangle &\le \frac{25}{2}\epsilon + \frac{1}{\sqrt{2}}\frac{5\epsilon}{\sqrt{2}}\\
    &=15 \epsilon. \notag
\end{align*}
Using this bound we can say
\begin{align*}
    \Re\langle\hat \Psi|\hat{A}_{34} A_{34}| \Psi\rangle &\ge \Re\langle\hat \Psi| \Psi\rangle-|\Re\langle\hat \Psi|(\mathds{1}-\hat{A}_{34} A_{34})| \Psi\rangle|\\
    &\ge 1- 7\epsilon-15\epsilon\\
    &= 1-22\epsilon.
\end{align*}
Therefore,
\begin{align*}
    \|A_{34}\ket{\Psi}-\hat{A}_{34}|\hat \Psi\rangle\|&=\sqrt{2-2 \Re\langle\hat \Psi| \Psi\rangle}\\
    &\le \sqrt{44\epsilon}.
\end{align*}

%\KMc{Derivation of alternative looser bound:
%We bound $\|A_{34}\ket{\Psi}-\hat A_{34}|\hat\Psi\rangle\|$
%by deriving a bound on $\langle\hat\Psi|A_{34}\hat{A}_{34}|\hat\Psi\rangle$.
%Using Eqs.~\eqref{eq: jordan angles}, \eqref{eq: ideal state}, and \eqref{eq: Col 2 fid proof}, we obtain
%\begin{align}
%    \langle\hat\Psi|A_{34}\hat{A}_{34}|\hat\Psi\rangle&= \sum_l p_l \cos{\phi_l}\notag\\
%    &\ge 1 - \frac{25}{2}\epsilon. \notag
%\end{align}
%It follows that
%\begin{equation} \notag
%    \sqrt{2-\langle\hat\Psi|(A_{34}\hat{A}_{34}+\hat{A}_{34}A_{34})|\hat\Psi\rangle}\le 5\sqrt{\epsilon}.
%\end{equation}
%Therefore, using the triangle inequality,
%\begin{align}
%    &\|A_{34}\ket{\Psi}-\hat A_{34}|\hat\Psi\rangle\|\notag\\
%    &\le\|A_{34}(\ket{\Psi}-|\hat\Psi\rangle)\|+\|(A_{34}-\hat A_{34})|\hat\Psi\rangle\|\notag\\
%    &= \|\ket{\Psi}-|\hat\Psi\rangle\|+\sqrt{2-\langle\hat\Psi|A_{34}\hat{A}_{34}+\hat{A}_{34}A_{34}|\hat\Psi\rangle}\notag\\
%    &\le \sqrt{14}\sqrt{\epsilon} + 5\sqrt{\epsilon}, \notag
%\end{align}
%where we have used the state fidelity bound, Eq.~\eqref{Fidelity: State}, for the first term.
%}

\end{document}